\documentclass[11pt]{article}

\usepackage[T1]{fontenc}
\usepackage[utf8]{inputenc}
\usepackage{cite}
\usepackage{times}
\usepackage[paperwidth=199.8mm,
paperheight=297mm,centering,hmargin=20mm,vmargin=20mm]{geometry}
\usepackage{authblk} 
\usepackage[bottom]{footmisc} 

\usepackage[titletoc,title]{appendix} 
\usepackage{changepage}

\usepackage{hyperref}
\hypersetup{colorlinks=true,citecolor=myblue,linkcolor=myblue,
filecolor=myblue,urlcolor=myblue,breaklinks=true}
\usepackage{url}

\usepackage[dvipsnames]{xcolor}
\usepackage{color}
\usepackage{framed}

\definecolor{shadecolor}{rgb}{0.9,0.9,0.9}
\definecolor{mylightgray}{RGB}{100,100,100}

\definecolor{myblue}{RGB}{0, 68, 116}
\definecolor{mycyan}{RGB}{0, 97, 91}
\definecolor{mygreen}{RGB}{2, 102, 1}
\definecolor{myorange}{RGB}{240, 102, 0}
\definecolor{myred}{RGB}{172, 23, 0}
\definecolor{mymagenta}{RGB}{140,16,73}

\usepackage{mathtools}
\usepackage{amsmath}
\usepackage{amssymb}
\usepackage[shortlabels]{enumitem}
\usepackage{graphicx,epsfig,latexsym,verbatim}
\usepackage{dsfont}
\usepackage{mathrsfs}

\usepackage{subcaption}
\usepackage{caption}
\usepackage{float}
\usepackage{tikz}
\usetikzlibrary{arrows.meta,positioning,calc,fit,matrix,backgrounds}
\usepackage{relsize}
\usepackage{pgfplots}
\pgfplotsset{compat=1.18}

\usepackage{amsthm}
\usepackage{tcolorbox}
\tcbuselibrary{skins}
\tcbuselibrary{breakable}

\definecolor{myblue}{RGB}{0,68,116}
\definecolor{mycyan}{RGB}{0,97,91}
\definecolor{mygreen}{RGB}{2,102,1}
\definecolor{myorange}{RGB}{240,102,0}
\definecolor{myred}{RGB}{172,23,0}

\newtheorem{theorem}{Theorem}
\newtheorem{lemma}[theorem]{Lemma}

\newtheorem{boxedtheorem}[theorem]{Theorem}
\newtheorem{boxedlemma}[theorem]{Lemma}
\newtheorem{boxedprop}[theorem]{Proposition}
\newtheorem{boxedcorollary}[theorem]{Corollary}
\newtheorem{boxeddefinition}[theorem]{Definition}
\newtheorem{boxedassumption}[theorem]{Assumption}
\newtheorem{boxedremark}[theorem]{Remark}

\newtheorem{boxedexample}[theorem]{Example}

\definecolor{theoremframe}{RGB}{31,78,121}
\definecolor{theoremback}{RGB}{247,250,253}
\definecolor{lemmaframe}{RGB}{35,105,86}
\definecolor{lemmaback}{RGB}{247,252,250}
\definecolor{propframe}{RGB}{102,73,139}
\definecolor{propback}{RGB}{251,249,253}
\definecolor{corollaryframe}{RGB}{165,96,18}
\definecolor{corollaryback}{RGB}{254,251,246}
\definecolor{definitionframe}{RGB}{112,112,112}
\definecolor{definitionback}{RGB}{248,248,248}

\tcbset{
  supplement result/.style={
    enhanced jigsaw,
    breakable,
    lines before break=5,
    sharp corners,
    boxrule=0.55pt,
    left=6pt,
    right=6pt,
    oversize=0pt,
    top=6pt,
    bottom=6pt,
    before skip=9pt plus 2pt minus 1pt,
    after skip=9pt plus 2pt minus 1pt
  }
}
\tcolorboxenvironment{boxedtheorem}{
  supplement result,
  colframe=myblue,
  colback=myblue!5!white
}
\tcolorboxenvironment{boxedlemma}{
  supplement result,
  colframe=definitionframe,
  colback=definitionback
}
\tcolorboxenvironment{boxedprop}{
  supplement result,
  colframe=definitionframe,
  colback=definitionback
}
\tcolorboxenvironment{boxedcorollary}{
  supplement result,
  colframe=definitionframe,
  colback=definitionback
}
\tcolorboxenvironment{boxeddefinition}{
  supplement result,
  colframe=definitionframe,
  colback=definitionback
}
\tcolorboxenvironment{boxedassumption}{
  supplement result,
  colframe=definitionframe,
  colback=definitionback
}
\tcolorboxenvironment{boxedremark}{
  supplement result,
  colframe=myorange!70!black,
  colback=myorange!6!white
}
\tcolorboxenvironment{boxmessage}{
  supplement result,
  colframe=myorange!70!black,
  colback=myorange!6!white
}
\newtcolorbox{boxnote}{
  supplement result,
  colframe=myorange!70!black,
  colback=myorange!6!white
}
\tcolorboxenvironment{boxedexample}{
  supplement result,
  colframe=definitionframe,
  colback=definitionback
}

\newenvironment{boxtheorem}[1][]
  {\begin{boxedtheorem}[#1]}
  {\end{boxedtheorem}}

\newenvironment{boxcorollary}[1][]
  {\begin{boxedcorollary}[#1]}
  {\end{boxedcorollary}}

\usepackage{colortbl}
\usepackage{pifont} 
\usepackage{booktabs}
\usepackage{makecell} 
\usepackage{diagbox}
\usepackage{multirow}

\usepackage{microtype}
\usepackage{needspace}
\newcommand{\nc}{\newcommand}
\nc{\rnc}{\renewcommand}

\nc{\<}{\langle}
\rnc{\>}{\rangle}
\nc{\bra}[1]{\langle#1|}
\nc{\ket}[1]{|#1\rangle}
\nc{\ketbra}[2]{|#1\rangle\!\langle#2|}
\nc{\braket}[2]{\langle#1|#2\rangle}
\nc{\braandket}[3]{\langle #1|#2|#3\rangle}
\nc{\proj}[1]{| #1\rangle\!\langle #1 |}
\nc{\avg}[1]{\langle#1\rangle}

\nc{\rank}{\operatorname{Rank}}
\nc{\id}{{\operatorname{id}}}
\nc{\iid}{{\operatorname{iid}}}
\nc{\supp}{{\operatorname{supp}}}
\nc{\smfrac}[2]{\mbox{$\frac{#1}{#2}$}}
\nc{\tr}{\operatorname{Tr}}
\nc{\coeff}{\operatorname{coeff}}
\nc{\ox}{\otimes}
\nc{\floor}[1]{\lfloor #1 \rfloor}
\nc{\trans}{\mathsf T}
\nc{\img}{\mathbf{i}}

\makeatletter

\newcommand{\solidgamearrowpicture@}[3]{%
    \tikz[x=1ex,y=1ex,line cap=round,line join=round]
      \draw[line width=#3,-{Triangle[length=0.64ex,width=0.68ex]}] (0,0) -- (#1,#2);%
}
\newcommand{\opengamearrowpicture@}[3]{%
    \tikz[x=1ex,y=1ex,line cap=round,line join=round]
      \draw[line width=#3,-{Triangle[open,length=0.64ex,width=0.68ex]}] (0,0) -- (#1,#2);%
}
\newsavebox{\solidgameuparrow@display}
\newsavebox{\solidgameuparrow@text}
\newsavebox{\solidgameuparrow@script}
\newsavebox{\solidgameuparrow@scriptscript}
\newsavebox{\solidgamedownarrow@display}
\newsavebox{\solidgamedownarrow@text}
\newsavebox{\solidgamedownarrow@script}
\newsavebox{\solidgamedownarrow@scriptscript}
\newsavebox{\opengameuparrow@display}
\newsavebox{\opengameuparrow@text}
\newsavebox{\opengameuparrow@script}
\newsavebox{\opengameuparrow@scriptscript}
\newsavebox{\opengamedownarrow@display}
\newsavebox{\opengamedownarrow@text}
\newsavebox{\opengamedownarrow@script}
\newsavebox{\opengamedownarrow@scriptscript}
\AtBeginDocument{%
  \sbox{\solidgameuparrow@display}{\solidgamearrowpicture@{0}{1.62}{0.58pt}}%
  \sbox{\solidgameuparrow@text}{\solidgamearrowpicture@{0}{1.50}{0.52pt}}%
  \sbox{\solidgameuparrow@script}{\solidgamearrowpicture@{0}{1.30}{0.44pt}}%
  \sbox{\solidgameuparrow@scriptscript}{\solidgamearrowpicture@{0}{1.12}{0.38pt}}%
  \sbox{\solidgamedownarrow@display}{\solidgamearrowpicture@{0}{-1.62}{0.58pt}}%
  \sbox{\solidgamedownarrow@text}{\solidgamearrowpicture@{0}{-1.50}{0.52pt}}%
  \sbox{\solidgamedownarrow@script}{\solidgamearrowpicture@{0}{-1.30}{0.44pt}}%
  \sbox{\solidgamedownarrow@scriptscript}{\solidgamearrowpicture@{0}{-1.12}{0.38pt}}%
  \sbox{\opengameuparrow@display}{\opengamearrowpicture@{0}{1.62}{0.58pt}}%
  \sbox{\opengameuparrow@text}{\opengamearrowpicture@{0}{1.50}{0.52pt}}%
  \sbox{\opengameuparrow@script}{\opengamearrowpicture@{0}{1.30}{0.44pt}}%
  \sbox{\opengameuparrow@scriptscript}{\opengamearrowpicture@{0}{1.12}{0.38pt}}%
  \sbox{\opengamedownarrow@display}{\opengamearrowpicture@{0}{-1.62}{0.58pt}}%
  \sbox{\opengamedownarrow@text}{\opengamearrowpicture@{0}{-1.50}{0.52pt}}%
  \sbox{\opengamedownarrow@script}{\opengamearrowpicture@{0}{-1.30}{0.44pt}}%
  \sbox{\opengamedownarrow@scriptscript}{\opengamearrowpicture@{0}{-1.12}{0.38pt}}%
}
\newcommand{\gamearrowbox@}[1]{%
  \mathrel{\vcenter{\hbox{%
    \usebox{#1}%
  }}}%
}
\DeclareRobustCommand{\soliduparrow}{%
  \mathchoice
    {\gamearrowbox@{\solidgameuparrow@display}}%
    {\gamearrowbox@{\solidgameuparrow@text}}%
    {\gamearrowbox@{\solidgameuparrow@script}}%
    {\gamearrowbox@{\solidgameuparrow@scriptscript}}%
}
\DeclareRobustCommand{\soliddownarrow}{%
  \mathchoice
    {\gamearrowbox@{\solidgamedownarrow@display}}%
    {\gamearrowbox@{\solidgamedownarrow@text}}%
    {\gamearrowbox@{\solidgamedownarrow@script}}%
    {\gamearrowbox@{\solidgamedownarrow@scriptscript}}%
}
\DeclareRobustCommand{\emptyuparrow}{%
  \mathchoice
    {\gamearrowbox@{\opengameuparrow@display}}%
    {\gamearrowbox@{\opengameuparrow@text}}%
    {\gamearrowbox@{\opengameuparrow@script}}%
    {\gamearrowbox@{\opengameuparrow@scriptscript}}%
}
\DeclareRobustCommand{\emptydownarrow}{%
  \mathchoice
    {\gamearrowbox@{\opengamedownarrow@display}}%
    {\gamearrowbox@{\opengamedownarrow@text}}%
    {\gamearrowbox@{\opengamedownarrow@script}}%
    {\gamearrowbox@{\opengamedownarrow@scriptscript}}%
}

\makeatother

\rnc{\liminf}{\mathop{\underline{\lim}}\displaylimits}
\rnc{\limsup}{\mathop{\overline{\lim}}\displaylimits}

\nc{\cA}{{\cal A}}
\nc{\cB}{{\cal B}}
\nc{\cC}{{\cal C}}
\nc{\cD}{{\cal D}}
\nc{\cE}{{\cal E}}
\nc{\cF}{{\cal F}}
\nc{\cG}{{\cal G}}
\nc{\cH}{{\cal H}}
\nc{\cI}{{\cal I}}
\nc{\cJ}{{\cal J}}
\nc{\cK}{{\cal K}}
\nc{\cL}{{\cal L}}
\nc{\cM}{{\cal M}}
\nc{\cN}{{\cal N}}
\nc{\cO}{{\cal O}}
\nc{\cP}{{\cal P}}
\nc{\cQ}{{\cal Q}}
\nc{\cR}{{\cal R}}
\nc{\cS}{{\cal S}}
\nc{\cT}{{\cal T}}
\nc{\cV}{{\cal V}}
\nc{\cU}{{\cal U}}
\nc{\cX}{{\cal X}}
\nc{\cY}{{\cal Y}}
\nc{\cZ}{{\cal Z}}
\nc{\cW}{{\cal W}}

\nc{\bp}{\boldsymbol{p}}
\nc{\bq}{\boldsymbol{q}}
\nc{\brho}{\boldsymbol{\rho}}
\nc{\bsigma}{\boldsymbol{\sigma}}
\nc{\bomega}{\boldsymbol{\omega}}
\nc{\bmu}{\boldsymbol{\mu}}
\nc{\bT}{\boldsymbol{T}}
\nc{\bS}{\boldsymbol{S}}

\nc{\RR}{{{\mathbb R}}}
\nc{\CC}{{{\mathbb C}}}
\nc{\FF}{{{\mathbb F}}}
\nc{\NN}{{{\mathbb N}}}
\nc{\ZZ}{{{\mathbb Z}}}
\nc{\QQ}{{{\mathbb Q}}}
\nc{\UU}{{{\mathbb U}}}
\nc{\EE}{{{\mathbb E}}}

\nc{\bH}{{\mathfrak{H}}}

\nc{\sK}{{{\mathscr{K}}}}
\nc{\sS}{{{\mathscr{S}}}}
\nc{\sT}{{{\mathscr{T}}}}
\nc{\sA}{{{\mathscr{A}}}}
\nc{\sB}{{{\mathscr{B}}}}
\nc{\sC}{{{\mathscr{C}}}}
\nc{\sE}{{{\mathscr{E}}}}
\nc{\sL}{{{\mathscr{L}}}}
\nc{\sG}{{{\mathscr{G}}}}
\nc{\sF}{{{\mathscr{F}}}}
\nc{\sI}{{{\mathscr{I}}}}
\nc{\sN}{{{\mathscr{N}}}}
\nc{\sM}{{{\mathscr{M}}}}

\nc{\Choi}{Choi-Jamio\l{}kowski }
\nc{\reg}{\infty}
\nc{\mx}{\text{\rm mx}}
\nc{\amo}{\text{\rm amo}}
\nc{\Renyi}{R\'{e}nyi }
\nc{\Stein}{\mathsf{Stein}}

\nc{\conv}{\operatorname{conv}}
\nc{\cvxset}{\mathscr{C}}

\nc{\RM}{{{\mathscr{R}}}}

\nc{\END}{\operatorname{End}}
\nc{\PERM}{\mathfrak{\sigma}}

\nc{\Cone}{\text{\rm Cone}}
\nc{\sep}{{\SEP}}

\nc{\DD}{{{\mathbb D}}}
\nc{\BS}{{\scriptscriptstyle \rm {BS}}}
\nc{\Sand}{{\scriptscriptstyle  \rm S}}
\nc{\Hypo}{{\scriptscriptstyle  \rm H}}
\nc{\Meas}{{\scriptscriptstyle \rm M}}
\nc{\Proj}{{{\scriptscriptstyle \rm P}}}
\nc{\KL}{{{\scriptscriptstyle \rm KL}}}
\nc{\IS}{{{\scriptscriptstyle \rm IS}}}

\nc{\suchthat}{\text{\rm s.t.}}

\nc{\pl}{{\scalebox{0.7}{+}}}
\nc{\HERM}{\mathscr{H}}
\nc{\PSD}{\HERM_{\pl}}
\nc{\PD}{\HERM_{\pl\pl}}
\nc{\density}{\mathscr{D}}
\nc{\subdensity}{\mathscr{D}_\bullet}

\nc{\polarPSD}[1]{{#1}_{\pl}^{\circ}}
\nc{\polarPSDre}[1]{{#1}_{\pl}^{\star}}
\nc{\polarPD}[1]{{#1}_{\pl\pl}^{\circ}}

\nc{\PPT}{\text{\rm PPT}}
\nc{\Rains}{\text{\rm Rains}}
\nc{\WD}{\text{\rm WD}}
\nc{\SEP}{\text{\rm SEP}}
\nc{\PSEP}{\text{\rm PSEP}}
\nc{\CPTP}{\text{\rm CPTP}}
\nc{\POVM}{\text{\rm POVM}}
\nc{\PVM}{\text{\rm PVM}}
\nc{\CP}{\text{\rm CP}}
\nc{\adv}{\text{\rm adv}}
\nc{\spec}{\text{\rm spec}}
\nc{\poly}{\text{\rm poly}}
\nc{\End}{\operatorname{End}}
\nc{\Par}{\operatorname{Par}}
\nc{\RNG}{\operatorname{RNG}}
\nc{\STAB}{\text{\rm STAB}}
\nc{\epi}{\boldsymbol{\operatorname{epi}}}
\nc{\op}{\boldsymbol{\operatorname{op}}}

\makeatletter
\newcommand*\rel@kern[1]{\kern#1\dimexpr\macc@kerna}
\newcommand*\widebar[1]{%
  \begingroup
  \def\mathaccent##1##2{%
    \rel@kern{0.8}%
    \overline{\rel@kern{-0.8}\macc@nucleus\rel@kern{0.2}}%
    \rel@kern{-0.2}%
  }%
  \macc@depth\@ne
  \let\math@bgroup\@empty \let\math@egroup\macc@set@skewchar
  \mathsurround\z@ \frozen@everymath{\mathgroup\macc@group\relax}%
  \macc@set@skewchar\relax
  \let\mathaccentV\macc@nested@a
  \macc@nested@a\relax111{#1}%
  \endgroup
}
\makeatother

\renewcommand{\mathbf}{\boldsymbol}

\renewcommand{\geq}{\geqslant}

\renewcommand{\leq}{\leqslant}

\DeclarePairedDelimiterX{\spr}[2]{\langle}{\rangle}{#1\delimsize\vert#2}
\DeclarePairedDelimiterX{\ceil}[1]{\lceil}{\rceil}{#1}
\DeclarePairedDelimiterX{\abs}[1]{\lvert}{\rvert}{#1}
\DeclarePairedDelimiterX{\norm}[1]{\lVert}{\rVert}{#1}
\DeclarePairedDelimiterX{\size}[1]{\lvert}{\rvert}{#1}
\DeclarePairedDelimiterX{\infdiv}[2]{(}{)}{#1\delimsize\Vert#2}
\DeclarePairedDelimiterX{\infdivc}[3]{(}{)}{#1\delimsize\Vert#2\delimsize\vert#3}
\DeclarePairedDelimiterX{\inner}[2]{\langle}{\rangle}{#1,#2}
\ExplSyntaxOn
\NewDocumentCommand{\multiadjustlimits}{m}
 {
  \group_begin:
  \multiadjustlimits_measure:n { #1 }
  \multiadjustlimits_print:n { #1 }
  \group_end:
 }

\tl_new:N  \l__multiadjustlimits_operator_tl
\tl_new:N  \l__multiadjustlimits_limit_tl

\cs_new_protected:Nn \multiadjustlimits_measure:n
 {
  \clist_map_function:nN { #1 } \__multiadjustlimits_measure:n
 }
\cs_new_protected:Nn \__multiadjustlimits_measure:n
 {
  \__multiadjustlimits_measure:NNn #1
 }
\cs_new_protected:Nn \__multiadjustlimits_measure:NNn
 {
  \tl_put_right:Nn \l__multiadjustlimits_operator_tl { #1 }
  \tl_put_right:Nn \l__multiadjustlimits_limit_tl { #3 }
 }

\cs_new_protected:Nn \multiadjustlimits_print:n
 {
  \clist_map_function:nN { #1 } \__multiadjustlimits_print:n
 }
\cs_new_protected:Nn \__multiadjustlimits_print:n
 {
  \__multiadjustlimits_print:NNn #1
 }
\cs_new_protected:Nn \__multiadjustlimits_print:NNn
 {
  \mathop { \vphantom{\l__multiadjustlimits_operator_tl} \mathopen{} #1 }
  \limits
  \sb{ \vphantom{\cramped{\l__multiadjustlimits_limit_tl}} #3 }
 }

\ExplSyntaxOff

\usepackage{framed}
\usepackage[dvipsnames]{xcolor} \usepackage{color}
\definecolor{shadecolor}{rgb}{0.9,0.9,0.9}

\newcommand{\Nch}{\cN}
\newcommand{\Tr}{\tr}
\newcommand{\epsfun}{f}

\title{\LARGE \textbf{The classical capacity of generalized amplitude-damping channels and its strong-converse exponent}}
\author[1]{Kun Fang \thanks{kunfang@cuhk.edu.cn}}
\affil[1]{\small School of Data Science, The Chinese University of Hong Kong, Shenzhen,\protect\\
Guangdong, 518172, China}
\date{\today}

\begin{document}
\maketitle

\begin{abstract}
The classical capacity of a quantum channel generally requires
regularization because input codewords may be entangled across channel uses.
We remove this regularization for generalized amplitude-damping
channels (GADCs) and determine their exact strong-converse exponent
throughout the full parameter range.
The exponent is a single-letter transform of the sandwiched
R\'enyi Holevo information and is attained by product input
codewords with collective decoding.
The key ingredient is a decomposition of the joint output into
positive operators whose triangular factors have identical scalar
Gram matrices after norm compression.
A sharp triangular Schatten-norm compression inequality then
establishes maximum output norm multiplicativity for diagonally
sandwiched GADCs tensored with arbitrary completely positive maps.
This yields additivity of the sandwiched R\'enyi Holevo information
at every order above one.
Consequently, the classical capacity equals the known one-copy
Holevo information and admits a scalar optimization.
Thus, entangled input codewords improve neither the capacity
nor the optimal exponential decay of decoding success above capacity.
\end{abstract}

{
\small
\tableofcontents
}

\section{Introduction}
\label{sec:intro}

The classical capacity of a quantum channel is the highest rate at
which classical information can be transmitted with vanishing error.
The Holevo--Schumacher--Westmoreland theorem expresses this capacity
as~\cite{Holevo1973,Holevo1998,SchumacherWestmoreland1997}
\begin{align}
 C(\Nch)=\lim_{n\to\infty}\frac{1}{n}\chi(\Nch^{\otimes n}),
 \label{eq:hsw}
\end{align}
where $\chi$ denotes the Holevo information.
The optimization allows input states entangled across arbitrarily
many channel uses, and this regularization is necessary in
general~\cite{Hastings2009}.
A central question is therefore whether a channel's one-copy
Holevo information already determines its classical capacity.
Even for qubit channels, removing the regularization requires
control of outputs produced by arbitrary entangled inputs.

The generalized amplitude-damping channel (GADC) models a two-level
system exchanging energy with a thermal
environment~\cite{KhatriSharmaWilde2020}.
It describes spin relaxation in quantum devices, including
superconducting circuits~\cite{ChirolliBurkard2008}, and has
applications to optical losses with low-temperature background
noise~\cite{ZouEtAl2017}.
We write $\cA_{\gamma,N}$ for the channel with damping probability
$\gamma$ and equilibrium occupation $N$; its action is given
in~\eqref{eq:gadc-action}.
The zero-temperature case $N=0$ is the amplitude-damping channel.

The one-copy Holevo information of the GADC is known through
a scalar optimization~\cite{LiZhenMaoFa2007,KhatriSharmaWilde2020}.
Earlier work developed computable bounds on the regularized
capacity~\cite{WangXieDuan2018,KhatriSharmaWilde2020,FangFawzi2021,FawziShayeghiTa2022}.
Recently, Tang, Zhu, Bai, and Wang~\cite{TangZhuBaiWang2026}
determined the classical capacity of amplitude damping through
a triangular block-matrix entropy inequality.
Their result establishes Holevo additivity with arbitrary other
channels for qubit channels admitting a pure output.
A GADC with $0<\gamma,N<1$ has no pure output,
so this criterion does not cover the finite-temperature case.

In this work, we resolve the regularization problem for the
full GADC family and determine its exact strong-converse exponent.
For a channel $\mathcal N$, let
$P_{\mathrm s}^\star(n,r;\mathcal N)$ denote the optimal average
decoding success probability over $n$ uses with at least $2^{nr}$ messages.
We allow arbitrary entangled input codewords and collective
decoding, but no entanglement assistance.
The strong-converse exponent is
\begin{align}
 E_{\mathrm{sc}}(r,\mathcal N)
 &=\lim_{n\to\infty}-\frac1n\log P_{\mathrm s}^\star(n,r;\mathcal N).
 \label{eq:strong-converse-exponent}
\end{align}
It describes the optimal exponential decay of decoding success
at rates above the classical capacity.
Cheng and Tomamichel~\cite[Theorem~5]{ChengTomamichel2026}
proved that the limit exists for every $r\geq0$ and
\begin{align}
 E_{\mathrm{sc}}(r,\mathcal N)
 &=\sup_{\alpha>1}\frac{\alpha-1}{\alpha}
 \left[r-\lim_{n\to\infty}\frac1n
              \widetilde\chi_\alpha(\mathcal N^{\otimes n})\right].
 \label{eq:general-exponent}
\end{align}
Thus, determining the exponent reduces to removing the
regularization of the sandwiched R\'enyi Holevo information
$\widetilde\chi_\alpha$, defined in Section~\ref{sec:setup}.
In contrast to the entropy argument of~\cite{TangZhuBaiWang2026},
this approach requires control at every order $\alpha>1$,
not only in the limit $\alpha\to1^+$.

\begin{boxtheorem}[Strong converse exponent]
\label{thm:strong-converse}
For every $\gamma,N\in[0,1]$, $\alpha>1$, and positive integer $n$,
\begin{align}
 \widetilde\chi_\alpha(\cA_{\gamma,N}^{\otimes n})
 &=n\widetilde\chi_\alpha(\cA_{\gamma,N}).
 \label{eq:renyi-additivity}
\end{align}
For every $r\geq0$, the limit in~\eqref{eq:strong-converse-exponent}
exists and satisfies
\begin{align}
 E_{\mathrm{sc}}(r,\cA_{\gamma,N})
 &=\sup_{\alpha>1}\frac{\alpha-1}{\alpha}
 [r-\widetilde\chi_\alpha(\cA_{\gamma,N})].
 \label{eq:exact-strong-converse-exponent}
\end{align}
\end{boxtheorem}

Taking $\alpha\to1^+$ in~\eqref{eq:renyi-additivity} gives
the following capacity formula.

\begin{boxcorollary}[Classical capacity]
\label{cor:capacity-additivity}
For every $\gamma,N\in[0,1]$,
\begin{align}
 C(\cA_{\gamma,N})
 &=\chi(\cA_{\gamma,N}).
 \label{eq:exact-capacity}
\end{align}
\end{boxcorollary}

The main task is to establish~\eqref{eq:renyi-additivity}. We follow the norm-based approach of Wilde, Winter, and
Yang~\cite{WildeWinterYang2014}, which relates the sandwiched
R\'enyi information radius to maximum output norms of sandwiched maps.
Our main technical result establishes the required multiplicativity
for GADCs sandwiched by arbitrary nonnegative diagonal operators
(Theorem~\ref{thm:filtered-multiplicativity}).
Phase covariance makes these diagonal operators sufficient for
the reference-state optimization.
The construction decomposes the joint output into a convex combination
of two positive operators with triangular factors.
Replacing each block of these factors by its Schatten norm
produces scalar matrices with identical Gram matrices.
A sharp triangular norm-compression inequality bounds each term
by the same one-copy output norm.
This matching property makes the argument applicable even when
the channel has no pure output.

The norm-compression proof extends the log-determinant method
of~\cite{TangZhuBaiWang2026} from entropy to fractional trace powers.
It establishes the previously missing range $2<p<4$,
complementing the known range
$p\geq4$~\cite{King2003,Audenaert2008}.

\section{Preliminaries}
\label{sec:setup}

All Hilbert spaces are finite dimensional.
A state is a positive semidefinite operator of trace one,
and $\operatorname{id}$ denotes the identity map.
All logarithms are to base two unless written as $\ln$.
We write $S(\rho)=-\Tr\rho\log\rho$ for the von Neumann entropy
and $h_2(x)=-x\log x-(1-x)\log(1-x)$ for the binary entropy,
with $0\log0=0$.

An $(n,M)$ code for $\mathcal N$ consists of states
$\rho_1,\ldots,\rho_M$ on the input space of $\mathcal N^{\otimes n}$
and a decoding positive operator-valued measure $\{D_m\}_{m=1}^M$
on its output space.
The decoding operators satisfy $D_m\geq0$ and $\sum_mD_m=I$.
The optimal average success probability at rate $r$ is
\begin{align}
 P_{\mathrm s}^\star(n,r;\mathcal N)
 &=\sup_{\substack{(n,M)\text{ codes}\\M\geq2^{nr}}}
   \frac1M\sum_{m=1}^M
   \Tr[\mathcal N^{\otimes n}(\rho_m)D_m].
 \label{eq:optimal-coding-success}
\end{align}
Inputs may be entangled across the $n$ channel uses,
and decoding may be collective.

\paragraph{The channel and its one-copy Holevo information.}
For $\gamma,N\in[0,1]$, the GADC acts as
\begin{align}
 \cA_{\gamma,N}:\quad
 \rho(q,z):=\begin{pmatrix}1-q&z\\z^*&q\end{pmatrix}
 \quad\longmapsto\quad
 \begin{pmatrix}
 1-t_{\gamma,N}(q)&\sqrt{1-\gamma}\,z\\
 \sqrt{1-\gamma}\,z^*&t_{\gamma,N}(q)
 \end{pmatrix},
 \label{eq:gadc-action}
\end{align}
where $q\in[0,1]$, $|z|^2\leq q(1-q)$, and
$t_{\gamma,N}(q)=\gamma N+(1-\gamma)q$.
The quantities $q$ and $t_{\gamma,N}(q)$ are the input and output
probabilities of $\ket1$.
The channel is phase covariant: for
$U_\theta=\ketbra00+e^{i\theta}\ketbra11$,
\begin{align}
 \cA_{\gamma,N}(U_\theta\rho U_\theta^\dagger)
 &=U_\theta\cA_{\gamma,N}(\rho)U_\theta^\dagger.
 \label{eq:gadc-phase-covariance}
\end{align}
For later use, define
\begin{align}
 \ket{\psi_q^\pm}
 &=\sqrt{1-q}\ket0\pm\sqrt q\ket1,
 \label{eq:capacity-achieving-states}\\
 \cA_{\gamma,N}(\ketbra{\psi_q^+}{\psi_q^+})
 &=\begin{pmatrix}
 1-t_{\gamma,N}(q)&\sqrt{(1-\gamma)q(1-q)}\\
 \sqrt{(1-\gamma)q(1-q)}&t_{\gamma,N}(q)
 \end{pmatrix},
 \label{eq:pure-output-matrix}\\
 D_{\gamma,N}(q)
 &:=\det\cA_{\gamma,N}(\ketbra{\psi_q^+}{\psi_q^+})
 \nonumber\\
 &=\gamma\big[(1-\gamma)(q-N)^2+N(1-N)\big].
 \label{eq:pure-output-determinant}
\end{align}
By phase covariance, every pure input with probability $q$ of $\ket1$
has an output unitarily equivalent to
$\cA_{\gamma,N}(\ketbra{\psi_q^+}{\psi_q^+})$.
Its entropy is $\epsfun(D_{\gamma,N}(q))$, where
\begin{align}
 \epsfun(d)
 &=h_2\!\left(\frac{1+\sqrt{1-4d}}{2}\right),
 \qquad 0\leq d\leq\frac14.
 \label{eq:entropy-function}
\end{align}

For a channel $\mathcal N$, its Holevo information is
\begin{align}
 \chi(\mathcal N)
 &=\sup_{\{p_x,\rho_x\}}
 \left[S\!\left(\sum_xp_x\mathcal N(\rho_x)\right)
       -\sum_xp_xS(\mathcal N(\rho_x))\right],
 \label{eq:holevo-information}
\end{align}
where the supremum is over finite input ensembles.
For the GADC, an equiprobable pair $\ket{\psi_q^\pm}$
is optimal for a suitable $q$~\cite{LiZhenMaoFa2007,KhatriSharmaWilde2020},
and hence
\begin{align}
 \chi(\cA_{\gamma,N})
 &=\max_{0\leq q\leq1}
 \left[h_2(t_{\gamma,N}(q))-\epsfun(D_{\gamma,N}(q))\right].
 \label{eq:holevo-gadc}
\end{align}
For this ensemble, both outputs have entropy
$\epsfun(D_{\gamma,N}(q))$, while their average is diagonal
with eigenvalues $1-t_{\gamma,N}(q)$ and $t_{\gamma,N}(q)$.

\paragraph{R\'enyi information and sandwiched norms.}
For $\alpha>1$, the sandwiched R\'enyi divergence and the
corresponding information radius of a channel are
\begin{align}
 \widetilde D_\alpha(\rho\|\sigma)
 &=\frac{1}{\alpha-1}\log\Tr
 \left[\left(\sigma^{\frac{1-\alpha}{2\alpha}}\rho
                   \sigma^{\frac{1-\alpha}{2\alpha}}\right)^\alpha\right],
 \label{eq:sandwiched-divergence}\\
 \widetilde\chi_\alpha(\mathcal N)
 &=\inf_{\sigma_B}\sup_{\rho_A}
 \widetilde D_\alpha(\mathcal N(\rho_A)\|\sigma_B).
 \label{eq:renyi-radius}
\end{align}
The optimizations range over states.
The divergence is infinite unless
$\operatorname{supp}\rho\subseteq\operatorname{supp}\sigma$.
When this inclusion holds, negative powers are taken on
the support of $\sigma$~\cite{MullerLennertEtAl2013,WildeWinterYang2014}.
It obeys data processing for every $\alpha>1$~\cite{FrankLieb2013}.
This radius equals the sandwiched R\'enyi Holevo information.
The identity was established for $1<\alpha\leq2$ by Wilde, Winter,
and Yang~\cite[Lemma~14]{WildeWinterYang2014}
and holds for all $\alpha>1$~\cite[Proposition~IV.2]{MosonyiOgawa2017}.
For a finite ensemble $\mathcal E=\{p_x,\rho_x\}$, let
$\omega_{XB}^{\mathcal E}=\sum_xp_x\ketbra{x}{x}\otimes\mathcal N(\rho_x)$
and $\omega_X^{\mathcal E}=\sum_xp_x\ketbra{x}{x}$.
Its sandwiched R\'enyi information is
$\widetilde I_\alpha(\mathcal E;\mathcal N)
=\inf_{\sigma_B}\widetilde D_\alpha
(\omega_{XB}^{\mathcal E}\|\omega_X^{\mathcal E}\otimes\sigma_B)$.
The ensemble representation of the channel information is then
\begin{align}
 \widetilde\chi_\alpha(\mathcal N)
 &=\sup_{\mathcal E}\widetilde I_\alpha(\mathcal E;\mathcal N).
 \label{eq:renyi-ensemble-information}
\end{align}

For $1\leq p<\infty$, write
$\|X\|_p=(\Tr|X|^p)^{1/p}$, where $|X|=(X^\dagger X)^{1/2}$;
$\|X\|_\infty$ denotes the operator norm.
The maximum output $p$-norm of a completely positive map $\Phi$ is
\begin{align}
 \nu_p(\Phi)
 &=\max_{\rho\geq0,\,\Tr\rho=1}\|\Phi(\rho)\|_p,
 \qquad 1\leq p\leq\infty.
 \label{eq:maximal-output-norm}
\end{align}
For $H\geq0$, we write $\Phi_H(X)=H\Phi(X)H$
for the associated \emph{sandwiched map}, following
\cite[Section~9]{WildeWinterYang2014}.
Neither $\Phi$ nor $\Phi_H$ is required to preserve trace.
For $H>0$, the sandwiched $p$-norm of
\cite[Section~5.2]{WildeWinterYang2014} is
$\|Y\|_{p,H^2}=\|HYH\|_p$, whose maximum over outputs $\Phi(\rho)$
is $\nu_p(\Phi_H)$.
For a strictly positive reference state $\sigma$, set
$H=\sigma^{(1-\alpha)/(2\alpha)}$.
The norm representation in~\cite[Eq.~(16)]{WildeWinterYang2014} gives
\begin{align}
 \sup_\rho\widetilde D_\alpha(\mathcal N(\rho)\|\sigma)
 &=\frac{\alpha}{\alpha-1}\log\nu_\alpha(\mathcal N_H),
 \label{eq:filtered-radius-identity}\\
 \widetilde\chi_\alpha(\mathcal N)
 &=\inf_{\sigma>0,\,\Tr\sigma=1}
   \frac{\alpha}{\alpha-1}\log\nu_\alpha
   (\mathcal N_{\sigma^{(1-\alpha)/(2\alpha)}}).
 \label{eq:renyi-norm-radius}
\end{align}
Restricting to product reference states and applying norm multiplicativity
then yields R\'enyi Holevo subadditivity, as in
\cite[Theorem~18]{WildeWinterYang2014}.

Strictly positive reference states suffice in~\eqref{eq:renyi-radius}.
For $\sigma_\delta=(1-\delta)\sigma+\delta I/d$ on a
$d$-dimensional output space and $0<\delta<1$,
\begin{align}
 \widetilde D_\alpha(\omega\|\sigma_\delta)
 &\leq\widetilde D_\alpha(\omega\|\sigma)-\log(1-\delta).
 \label{eq:faithful-renyi-reference}
\end{align}
When the right-hand side is finite, $\omega$ is supported on
$\operatorname{supp}\sigma$.
On this support, the commuting states satisfy
$\sigma_\delta^{(1-\alpha)/\alpha}
\leq(1-\delta)^{(1-\alpha)/\alpha}\sigma^{(1-\alpha)/\alpha}$.
Sandwiching by $\omega^{1/2}$ and applying eigenvalue monotonicity
to the trace-power expression proves the bound.
Since it is uniform in $\omega$, taking $\delta\to0^+$
leaves the optimized radius unchanged.

\section{Additivity of Holevo R\'enyi information}
\label{sec:filtered-norms}

We prove the additivity statement in
Theorem~\ref{thm:strong-converse} through three steps.
First, we establish a Schatten-norm compression inequality
for triangular block matrices.
We then apply it to prove multiplicativity of the maximum output norm
for diagonally sandwiched GADCs.
Finally, phase covariance and the norm
representation~\eqref{eq:renyi-norm-radius} convert this multiplicativity
into additivity of Holevo R\'enyi information,
following the framework of~\cite{WildeWinterYang2014}.

\Needspace{12\baselineskip}
\subsection{Schatten-norm compression}
\label{sec:norm-compression}

We first establish the triangular norm-compression inequality
used to control the factors in the output decomposition.

\begin{boxtheorem}[Schatten-norm compression]
\label{thm:triangular-norm}
For every $2\leq p\leq\infty$ and rectangular matrices $U,V,W$
of compatible sizes,
\begin{align}
 \left\|\begin{pmatrix}U&0\\V&W\end{pmatrix}\right\|_p
 &\leq
 \left\|\begin{pmatrix}\|U\|_p&0\\\|V\|_p&\|W\|_p\end{pmatrix}\right\|_p.
 \label{eq:triangular-norm}
\end{align}
The constant one is sharp for every prescribed triple of nonnegative
block norms and every compatible choice of positive block dimensions.
\end{boxtheorem}

The main step is to establish the previously missing range $2<p<4$.
Polar decomposition reduces the norm comparison to an inequality
for a trace-power difference.
Joint convexity allows us to remove off-diagonal entries from two
of its arguments; Jensen's and H\"older's inequalities then give
the scalar comparison.
The convexity argument extends the log-determinant method
of~\cite[Appendix~A]{TangZhuBaiWang2026} from entropy to trace powers.
For $p\geq4$, the result follows from Audenaert's two-row
compression inequality~\cite[Section~4.4]{Audenaert2008},
based on King's positive-block inequality~\cite{King2003}.

\begin{proof}
Consider first square blocks of the same size and $2<p<4$.
Set $s=p/2$ and $T=\left(\begin{smallmatrix}U&0\\V&W\end{smallmatrix}\right)$.

\emph{Step 1: Gram-matrix reduction.}
For positive semidefinite $P,Q,R$ of a common dimension $d$,
define the positive semidefinite block matrix
\begin{align}
 Z(P,Q,R)&=
 \begin{pmatrix}P+Q&\sqrt Q\sqrt R\\\sqrt R\sqrt Q&R\end{pmatrix}.
 \label{eq:norm-gram}
\end{align}
Choose polar decompositions $V=U_V\sqrt{V^\dagger V}$ and
$W=U_W\sqrt{W^\dagger W}$, extending the polar factors to unitaries.
With $J=U_V^\dagger U_W$, set
$P=U^\dagger U$, $Q=V^\dagger V$, and $R=JW^\dagger WJ^\dagger$.
Then
\begin{align}
 (I\oplus J)T^\dagger T(I\oplus J^\dagger)&=Z(P,Q,R).
 \label{eq:norm-polar-reduction}
\end{align}
To compare the block matrix with its scalar compression,
define the trace-power difference
\begin{align}
 F_s(P,Q,R)&=\Tr P^s+\Tr Q^s+\Tr R^s-\Tr Z(P,Q,R)^s.
 \label{eq:power-difference}
\end{align}
We use the same notation for nonnegative scalar arguments;
the scalar function is homogeneous of degree $s$.
The desired norm comparison reduces to
\begin{align}
 F_s(P,Q,R)&\geq F_s(\|P\|_s,\|Q\|_s,\|R\|_s).
 \label{eq:schatten-scalar-reduction}
\end{align}
Indeed, replacing $P,Q,R$ by their Schatten $s$-norms preserves
the first three terms of $F_s$.
Thus, a lower bound on $F_s$ gives an upper bound on
the Schatten $s$-norm of the Gram matrix.
The identities $\|P\|_s=\|U\|_p^2$, $\|Q\|_s=\|V\|_p^2$,
and $\|R\|_s=\|W\|_p^2$ identify the required scalar comparison.

\emph{Step 2: Scalar compression.}
First suppose $Q>0$, and choose an eigenbasis with eigenvalues $q_i$.
Let $x_i,z_i$ be the diagonal entries of $P,R$ in this basis,
and set $y=\|Q\|_s$.
We have
\begin{align}
 F_s(P,Q,R)
 &\overset{\mathrm{(a)}}{\geq}\sum_i F_s(x_i,q_i,z_i)\nonumber\\
 &=y^s\sum_i\frac{q_i^s}{y^s}
 F_s\!\left(\frac{x_i}{q_i},1,\frac{z_i}{q_i}\right)\nonumber\\
 &\overset{\mathrm{(b)}}{\geq} y^s F_s\!\left(
 \frac{\Tr(Q^{s-1}P)}{y^s},1,
 \frac{\Tr(Q^{s-1}R)}{y^s}\right)\nonumber\\
 &\overset{\mathrm{(c)}}{\geq} y^s F_s\!\left(\frac{\|P\|_s}{y},1,\frac{\|R\|_s}{y}\right)
 \nonumber\\
 &=F_s(\|P\|_s,\|Q\|_s,\|R\|_s).
 \label{eq:power-scalar-compression}
\end{align}
For step~(a), average simultaneous conjugations of $P,Q,R$
by diagonal sign unitaries in the chosen eigenbasis of $Q$.
This fixes $Q$ and removes the off-diagonal entries of $P,R$.
The resulting Gram matrix is a direct sum of scalar $2\times2$ blocks.
Unitary invariance and joint convexity in $(P,R)$,
proved in Lemma~\ref{lem:power-convexity}
(Appendix~\ref{sec:power-convexity}), give the inequality.
The next equality uses homogeneity, and step~(b) applies
Jensen's inequality with weights $q_i^s/y^s$, which sum to one.
For step~(c), H\"older's inequality gives
$\Tr(Q^{s-1}P)\leq y^{s-1}\|P\|_s$ and similarly for $R$.
Scalar monotonicity from the same lemma gives the inequality
in the stated direction.
The final equality again uses homogeneity.
Continuity, with $Q$ replaced by $Q+\epsilon I$,
extends~\eqref{eq:schatten-scalar-reduction} to $Q\geq0$.

Write $a=\|U\|_p$, $b=\|V\|_p$, and $c=\|W\|_p$.
Combining the two steps yields
\begin{align}
 \|T\|_p^2
 &=\|Z(P,Q,R)\|_s
\leq\left\|\begin{pmatrix}a^2+b^2&bc\\bc&c^2\end{pmatrix}\right\|_s
=\left\|\begin{pmatrix}a&0\\b&c\end{pmatrix}\right\|_p^2.
 \label{eq:norm-compression-proof}
\end{align}
This proves the claimed inequality throughout $2<p<4$.
For $4\leq p<\infty$, it is the triangular specialization
of~\cite[Section~4.4]{Audenaert2008}; at $p=2$, it is the Frobenius identity.

Zero padding embeds compatible rectangular blocks into square blocks
of a common size, preserving the nonzero singular values of every
block and of the full matrix.
The limit $p\to\infty$ gives the operator-norm endpoint.
Finally, choose unit vectors $\ket{u_1},\ket{u_2},\ket{v_1},\ket{v_2}$
in the corresponding row and column spaces, and set
$U=a\ketbra{u_1}{v_1}$, $V=b\ketbra{u_2}{v_1}$,
and $W=c\ketbra{u_2}{v_2}$.
The full matrix is then an isometric embedding of its scalar compression,
proving sharpness for every $a,b,c\geq0$.
\end{proof}

\subsection{Multiplicativity of maximum output norm}
\label{sec:channel-decomposition}

The compression inequality now yields multiplicativity of the maximum
output norm for the GADC sandwiched by any nonnegative diagonal operator.

\begin{boxtheorem}[Sandwiched norm multiplicativity]
\label{thm:filtered-multiplicativity}
For every $\gamma,N\in[0,1]$, every nonnegative diagonal $H$,
every finite-dimensional completely positive map $\Phi$,
and every $1\leq p\leq\infty$,
\begin{align}
 \nu_p((\cA_{\gamma,N})_H\otimes\Phi)&=\nu_p((\cA_{\gamma,N})_H)\nu_p(\Phi).
 \label{eq:filtered-multiplicativity}
\end{align}
The input and output dimensions of $\Phi$ are arbitrary,
and neither map is required to preserve trace.
\end{boxtheorem}

The proof decomposes the joint output into a convex combination
of two positive operators with triangular factors.
The decomposition parameter is chosen from the block norms so that
both factors have the same scalar Gram matrix after compression.
Theorem~\ref{thm:triangular-norm} and the triangle inequality then
reduce the joint output norm to a one-copy output norm
times a scalar controlled by $\nu_p(\Phi)$.

\begin{proof}
First assume $1\leq p<\infty$ and $0<\gamma,N<1$.
Convexity of the Schatten norm reduces the optimization to
normalized pure inputs
$\ket{\psi}=\ket0\otimes\ket{u}+\ket1\otimes\ket{v}$,
where $\langle u|u\rangle+\langle v|v\rangle=1$.
Since $\Phi$ is completely positive, the intermediate output is
positive semidefinite and admits the factorization
\begin{align}
 R&:=(\operatorname{id}\otimes\Phi)(\ketbra{\psi}{\psi})
 =\begin{pmatrix}
 \Phi(\ketbra{u}{u})&\Phi(\ketbra{u}{v})\\
 \Phi(\ketbra{v}{u})&\Phi(\ketbra{v}{v})
 \end{pmatrix}
 =\begin{pmatrix}X\\Y\end{pmatrix}
   \begin{pmatrix}X^\dagger&Y^\dagger\end{pmatrix}.
 \label{eq:partner-common-factor}
\end{align}
Set $m=\|X\|_{2p}^2+\|Y\|_{2p}^2$.
If $m=0$, the output vanishes and the desired bound is immediate.
Otherwise, set $q=\|Y\|_{2p}^2/m$ and first assume
that both $X$ and $Y$ are nonzero, so $0<q<1$.

\emph{Step 1: Decompose the output into triangular factors.}
Write $t=t_{\gamma,N}(q)$ and $D=D_{\gamma,N}(q)$.
The interior assumptions ensure $0<t<1$ and $D>0$.
Set $\lambda_0=\gamma(1-N)qt/D$ and
$\lambda_1=\gamma N(1-q)(1-t)/D$.
The output admits the decomposition
\begin{align}
 (\cA_{\gamma,N}\otimes\operatorname{id})(R)
 &=\sum_{j=0}^1\lambda_jM_jM_j^\dagger,
 \label{eq:thermal-kraus-factors}
\end{align}
where
\begin{align}
 M_0&=
 \begin{pmatrix}
 \sqrt{\dfrac{(1-\gamma)q}{t}}\,X&
 \sqrt{\dfrac{D}{qt}}\,Y\\[0.2cm]
 \sqrt{\dfrac{t}{q}}\,Y&0
 \end{pmatrix},
 &
 M_1&=
 \begin{pmatrix}
 \sqrt{\dfrac{1-t}{1-q}}\,X&0\\[0.2cm]
 \sqrt{\dfrac{(1-\gamma)(1-q)}{1-t}}\,Y&
 \sqrt{\dfrac{D}{(1-q)(1-t)}}\,X
 \end{pmatrix}.
 \label{eq:thermal-branch-matrices}
\end{align}
Direct block multiplication verifies~\eqref{eq:thermal-kraus-factors}.
The identity $D=\gamma(1-N)qt+\gamma N(1-q)(1-t)$
shows that the positive weights sum to one.
The factor $M_1$ is triangular, and $M_0$ becomes triangular
after exchanging its block rows.

Write $H=h_0\ketbra00+h_1\ketbra11$ with $h_0,h_1\geq0$.
For $j=0,1$, let $B_j$ be the scalar matrix obtained by replacing
each block of $(H\otimes I)M_j$ by its Schatten $2p$-norm.
The choice of $q$ ensures that both scalar Gram matrices coincide:
\begin{align}
 B_jB_j^\dagger
 &=m(\cA_{\gamma,N})_H(\ketbra{\psi_q^+}{\psi_q^+}),
 \qquad j=0,1.
 \label{eq:norm-matching-grams}
\end{align}
Indeed, substitute $\|X\|_{2p}=\sqrt{m(1-q)}$ and
$\|Y\|_{2p}=\sqrt{mq}$, and use
$D=t(1-t)-(1-\gamma)q(1-q)$.
The result is precisely the matrix in~\eqref{eq:pure-output-matrix},
multiplied by $m$ and sandwiched by $H$.

\emph{Step 2: Norm compression via Theorem~\ref{thm:triangular-norm}.}
We obtain the norm comparison
\begin{align}
 \|((\cA_{\gamma,N})_H\otimes\operatorname{id})(R)\|_p
 &\overset{\mathrm{(a)}}{=}
 \left\|\sum_{j=0}^1\lambda_j
       (H\otimes I)M_jM_j^\dagger(H\otimes I)\right\|_p
 \nonumber\\
 &\overset{\mathrm{(b)}}{\leq}
 \sum_{j=0}^1\lambda_j\|(H\otimes I)M_j\|_{2p}^2
 \nonumber\\
 &\overset{\mathrm{(c)}}{\leq}
 \sum_{j=0}^1\lambda_j\|B_j\|_{2p}^2
 \nonumber\\
 &=\sum_{j=0}^1\lambda_j\|B_jB_j^\dagger\|_p
 \nonumber\\
 &\overset{\mathrm{(d)}}{=}m\|(\cA_{\gamma,N})_H(\ketbra{\psi_q^+}{\psi_q^+})\|_p.
 \label{eq:norm-branch-comparison}
\end{align}
Step~(a) substitutes the output decomposition~\eqref{eq:thermal-kraus-factors}.
Step~(b) uses the triangle inequality and
$\|AA^\dagger\|_p=\|A\|_{2p}^2$.
Step~(c) is Theorem~\ref{thm:triangular-norm} at order $2p\geq2$.
For $M_0$, apply the same block-row permutation to
$(H\otimes I)M_0$ and $B_0$; Schatten norms are unchanged.
The nonnegative diagonal operator $H$ preserves the zero blocks.
The next equality again uses the relation between the two Schatten norms.
Step~(d) uses the common Gram matrix in~\eqref{eq:norm-matching-grams}
and the normalization $\lambda_0+\lambda_1=1$.

It remains to bound the scalar factor $m$ and the one-copy output norm:
\begin{align}
 \|((\cA_{\gamma,N})_H\otimes\Phi)(\ketbra{\psi}{\psi})\|_p
 &\overset{\mathrm{(e)}}{\leq}m\|(\cA_{\gamma,N})_H(\ketbra{\psi_q^+}{\psi_q^+})\|_p
 \nonumber\\
 &\overset{\mathrm{(f)}}{=}
 \bigl[\|\Phi(\ketbra{u}{u})\|_p+\|\Phi(\ketbra{v}{v})\|_p\bigr]
 \cdot\|(\cA_{\gamma,N})_H(\ketbra{\psi_q^+}{\psi_q^+})\|_p
 \nonumber\\
 &\overset{\mathrm{(g)}}{\leq}
 (\langle u|u\rangle+\langle v|v\rangle)\nu_p(\Phi)\nu_p((\cA_{\gamma,N})_H)
 \nonumber\\
 &=\nu_p(\Phi)\nu_p((\cA_{\gamma,N})_H).
 \label{eq:filtered-partner-proof}
\end{align}
Step~(e) is~\eqref{eq:norm-branch-comparison}, using the definition of $R$.
Step~(f) uses $XX^\dagger=\Phi(\ketbra{u}{u})$,
$YY^\dagger=\Phi(\ketbra{v}{v})$, and
$\|X\|_{2p}^2=\|XX^\dagger\|_p$.
Step~(g) follows from homogeneity and the definition of the maximum output
norm, since $\ket{\psi_q^+}$ is normalized.
The final equality uses the normalization of the input state.

If $q=0$ or $q=1$, exactly one of $X,Y$ is nonzero.
Then $R$ is a computational-basis projection tensored with
a positive matrix, and the comparison between the first and last
expressions in~\eqref{eq:norm-branch-comparison} holds with equality.
For fixed $p<\infty$, continuity in $\gamma,N$ extends the bound
to the boundary of $[0,1]^2$, without evaluating singular factorization formulas.
Zero diagonal entries of $H$ are already permitted.
Finally, taking $p\to\infty$ gives the operator-norm endpoint.

Maximizing over inputs proves the upper bound.
Product inputs give the reverse inequality because Schatten
norms multiply on tensor products, completing the proof.
\end{proof}

For integer $p$, sandwiched norm multiplicativity also follows from
the entrywise-positive map theorem of King, Nathanson, and
Ruskai~\cite[Theorem~3]{KingNathansonRuskai2005}.
The GADC sandwiched by any nonnegative diagonal operator
satisfies that theorem's matrix-unit positivity condition.
Theorem~\ref{thm:filtered-multiplicativity} extends this conclusion
to every real order $p\geq1$.
In particular, the range $1<p<2$ is essential for taking
the order-one limit that determines the capacity.

\subsection{Additivity of Holevo R\'enyi information}
\label{sec:renyi-additivity}

We now apply Theorem~\ref{thm:filtered-multiplicativity}
to the information-radius representation.
Phase covariance reduces the one-copy reference optimization to
diagonal states, and product reference states give the tensor-power
upper bound.
Product input ensembles establish the reverse inequality.

\begin{proof}[Proof of Theorem~\ref{thm:strong-converse}]
Fix $\alpha>1$.

\emph{Step 1: Diagonal reference states.}
Let $Z=\operatorname{diag}(1,-1)$ and
$\overline\sigma=(\sigma+Z\sigma Z)/2$.
For every output state $\omega$ and strictly positive $\sigma$,
\begin{align}
 \widetilde D_\alpha(\omega\|\overline\sigma)
 &\leq\frac{1}{\alpha-1}\log
 \frac{2^{(\alpha-1)\widetilde D_\alpha(\omega\|\sigma)}
       +2^{(\alpha-1)\widetilde D_\alpha(Z\omega Z\|\sigma)}}{2}
 \nonumber\\
 &\leq\max\{\widetilde D_\alpha(\omega\|\sigma),
             \widetilde D_\alpha(Z\omega Z\|\sigma)\}.
 \label{eq:renyi-reference-twirl}
\end{align}
The first inequality follows from data processing under discarding
a uniformly distributed binary flag.
Before the flag is discarded, the first state has $\omega$ in both
blocks, while the reference has $\sigma$ and $Z\sigma Z$.
Evaluating their sandwiched R\'enyi divergence and using unitary
invariance gives the logarithmic average in the first line.
The second inequality bounds the average of the two exponential
terms by their maximum.

The covariance relation~\eqref{eq:gadc-phase-covariance} at
$\theta=\pi$ makes the output set invariant under conjugation by $Z$.
Taking the supremum over outputs in~\eqref{eq:renyi-reference-twirl}
therefore shows that diagonalizing the reference cannot
increase its radius.
Together with~\eqref{eq:faithful-renyi-reference}
and~\eqref{eq:filtered-radius-identity}, this gives
\begin{align}
 \widetilde\chi_\alpha(\cA_{\gamma,N})
 &=\inf_{0<s<1}\sup_\rho
       \widetilde D_\alpha(\cA_{\gamma,N}(\rho)\|\sigma_s)
 \nonumber\\
 &=\inf_{0<s<1}\frac{\alpha}{\alpha-1}
       \log\nu_\alpha((\cA_{\gamma,N})_{H_{\alpha,s}}),
 \label{eq:diagonal-filtered-radius}
\end{align}
where $\sigma_s=\operatorname{diag}(1-s,s)$ and
$H_{\alpha,s}=\sigma_s^{(1-\alpha)/(2\alpha)}$.

\emph{Step 2: R\'enyi Holevo additivity.}
Using $\sigma_s^{\otimes n}$ as a trial reference, we obtain
\begin{align}
 \widetilde\chi_\alpha(\cA_{\gamma,N}^{\otimes n})
 &\leq\inf_{0<s<1}\sup_{\rho_n}
       \widetilde D_\alpha
       (\cA_{\gamma,N}^{\otimes n}(\rho_n)\|\sigma_s^{\otimes n})
 \nonumber\\
 &=\inf_{0<s<1}\frac{\alpha}{\alpha-1}
       \log\nu_\alpha((\cA_{\gamma,N})_{H_{\alpha,s}}^{\otimes n})
 \nonumber\\
 &=n\inf_{0<s<1}\frac{\alpha}{\alpha-1}
       \log\nu_\alpha((\cA_{\gamma,N})_{H_{\alpha,s}})
 \nonumber\\
 &=n\widetilde\chi_\alpha(\cA_{\gamma,N}).
 \label{eq:renyi-additivity-upper}
\end{align}
The inequality restricts the reference optimization,
without restricting the input $\rho_n$.
The first equality is the norm representation~\eqref{eq:filtered-radius-identity}.
The second applies Theorem~\ref{thm:filtered-multiplicativity}
repeatedly, with a sandwiched completely positive map as the other factor.
The last equality is~\eqref{eq:diagonal-filtered-radius}.

For the reverse direction, every finite input ensemble $\mathcal E$
satisfies
\begin{align}
 \widetilde\chi_\alpha(\cA_{\gamma,N}^{\otimes n})
 &\geq\widetilde I_\alpha(\mathcal E^{\otimes n};\cA_{\gamma,N}^{\otimes n})
 =n\widetilde I_\alpha(\mathcal E;\cA_{\gamma,N}).
 \label{eq:renyi-product-ensemble}
\end{align}
The inequality uses the ensemble representation~\eqref{eq:renyi-ensemble-information}.
The equality follows from~\cite[Lemma~IV.8]{MosonyiOgawa2017},
which establishes additivity of the ensemble information for product
ensembles even when the reference optimization allows correlated states.
Taking the supremum over $\mathcal E$ gives the lower bound
and completes the proof of~\eqref{eq:renyi-additivity}.

Substituting~\eqref{eq:renyi-additivity} into~\eqref{eq:general-exponent}
gives the exponent formula~\eqref{eq:exact-strong-converse-exponent}.
\end{proof}

\begin{proof}[Proof of Corollary~\ref{cor:capacity-additivity}]
For each fixed positive integer $n$, take $\alpha\to1^+$
in~\eqref{eq:renyi-additivity}.
The convergence of sandwiched R\'enyi Holevo information to Holevo
information~\cite[Lemma~IV.6]{MosonyiOgawa2017} gives
\begin{align}
 \chi(\cA_{\gamma,N}^{\otimes n})
 &=n\chi(\cA_{\gamma,N}).
 \label{eq:holevo-tensor-additivity}
\end{align}
Substituting this identity into the Holevo--Schumacher--Westmoreland
formula~\eqref{eq:hsw} proves~\eqref{eq:exact-capacity}.
\end{proof}

\paragraph{Acknowledgements.}
K.F. is supported in part by the National Natural Science Foundation of China (Grants No. 92470113 and 12404569), the Shenzhen Science and Technology Program (Grants No. QNXMB20250701091826036 and JCYJ20240813113519025), the Shenzhen Fundamental Research Program (Grant No. JCYJ20241202124023031), the General R\&D Projects of 1+1+1 CUHK-CUHK(SZ)-GDST Joint Collaboration Fund (Grant No. GRDP2025-022), the Guangdong Provincial Quantum Science Strategic Initiative (Grant No. GDZX2503001), and the University Development Fund (Grant No. UDF01003565). 

We acknowledge that OpenAI GPT-6 Astra was used to explore proof strategies and support the writing of the manuscript. We
have verified AI-assisted material
and take full responsibility for the content of this work.

\begingroup
\bibliographystyle{alpha}
\bibliography{references}
\endgroup

\appendix
\section{Convexity of the trace-power difference}
\label{sec:power-convexity}

The function $F_s$ in~\eqref{eq:power-difference} is the
trace-power analogue of the entropy difference studied in
\cite[Appendix~A]{TangZhuBaiWang2026}.
A positive integral representation in terms of the same
log-determinant functional yields the convexity needed in
Section~\ref{sec:norm-compression}.

\begin{lemma}[Convexity and scalar monotonicity]
\label{lem:power-convexity}
For $1<s<2$ and fixed $Q\geq0$, $F_s(P,Q,R)$
is jointly convex over positive semidefinite $P,R$.
For fixed $y\geq0$, the scalar function $F_s(x,y,z)$
is nonincreasing in each of $x,z\geq0$.
\end{lemma}

\begin{proof}
For $u>0$, define
\begin{align}
 G_u&=\begin{pmatrix}P+uI&0\\0&R+uI\end{pmatrix},
\qquad
 J_u=\begin{pmatrix}0&\sqrt{uQ}\\-\sqrt{uQ}&0\end{pmatrix}.
 \label{eq:power-schur-matrices}
\end{align}
With $\kappa_s=-s\sin(\pi s)/\pi>0$, we have
\begin{align}
 F_s(P,Q,R)
 &=\kappa_s\int_0^\infty u^{s-1}
 \left[\ln\frac{\det(G_u+J_u)}{\det G_u}
       -\ln\det(I+Q/u)\right]\,\mathrm du.
 \label{eq:power-convex-integral}
\end{align}
To see this, apply the scalar identity
\begin{align}
 x^s&=\kappa_s\int_0^\infty
 u^{s-1}\bigl[x/u-\ln(1+x/u)\bigr]\,\mathrm du,
 \qquad x\geq0,
 \label{eq:fractional-power-integral}
\end{align}
to the eigenvalues of $P\oplus Q\oplus R$ and $Z(P,Q,R)\oplus0_d$.
Their equal traces cancel the linear terms, and their equal
dimensions match the determinant normalizations.
The Schur-complement identity
$\det(Z(P,Q,R)+uI)=\det(G_u+J_u)$
from~\cite[Eqs.~(A.7)--(A.9)]{TangZhuBaiWang2026}
then gives~\eqref{eq:power-convex-integral}.
Since $Z(P,Q,R)+uI>0$, this identity also ensures
$\det(G_u+J_u)>0$, so its logarithm is well defined.
The logarithmic bracket is $O(1+|\ln u|)$ as $u\to0^+$
and $O(u^{-2})$ as $u\to\infty$.
The integral therefore converges for $1<s<2$, even when $P,Q,R$
are singular.

For fixed $Q$, the matrix $J_u$ is fixed and skew-Hermitian,
while $G_u$ depends affinely on $(P,R)$.
By~\cite[Lemma~A.1]{TangZhuBaiWang2026}, the first logarithmic term
is therefore jointly convex in $(P,R)$; the second is constant.
Integration against the positive weight $\kappa_s u^{s-1}$
proves joint convexity.
For scalar arguments $(x,y,z)$, the determinant ratio is
\begin{align}
 \frac{\det(G_u+J_u)}{\det G_u}
 &=1+\frac{uy}{(x+u)(z+u)}.
 \label{eq:power-scalar-monotonicity}
\end{align}
This ratio is nonincreasing in $x$ and $z$, while the remaining
logarithmic term is independent of them.
The same integral representation therefore proves scalar monotonicity.
\end{proof}

\section{Scalar evaluation of the R\'enyi Holevo information}
\label{sec:scalar-renyi}

We express the sandwiched R\'enyi Holevo information as an
optimization over the reference parameter $s$ and the input parameter $q$.
The objective is explicit, and the order $\inf_s\max_q$ is retained.

\begin{lemma}[Scalar R\'enyi radius]
\label{lem:scalar-renyi-radius}
Fix $\gamma,N\in[0,1]$ and $\alpha>1$, and let
$\sigma_s=\operatorname{diag}(1-s,s)$ for $0<s<1$.
For the states $\ket{\psi_q^+}$ defined in~\eqref{eq:capacity-achieving-states},
\begin{align}
 \widetilde\chi_\alpha(\cA_{\gamma,N})
 &=\inf_{0<s<1}\max_{0\leq q\leq1}
       \widetilde D_\alpha(\cA_{\gamma,N}(\ketbra{\psi_q^+}{\psi_q^+})\|\sigma_s)\nonumber\\
 &=\frac{1}{\alpha-1}\log\!\left[
 \inf_{0<s<1}\max_{0\leq q\leq1}
 \{\lambda_+(\alpha,s,q)^\alpha+\lambda_-(\alpha,s,q)^\alpha\}\right].
 \label{eq:scalar-renyi-radius}
\end{align}
The eigenvalues of the sandwiched output are given by
\begin{align}
 T_\alpha(s,q)
 &=[1-t_{\gamma,N}(q)](1-s)^{\frac{1-\alpha}{\alpha}}
   +t_{\gamma,N}(q)s^{\frac{1-\alpha}{\alpha}},\nonumber\\
 \lambda_\pm(\alpha,s,q)
 &=\frac{T_\alpha(s,q)\pm
 \sqrt{T_\alpha(s,q)^2-4D_{\gamma,N}(q)[s(1-s)]^{\frac{1-\alpha}{\alpha}}}}{2}.
 \label{eq:scalar-renyi-eigenvalues}
\end{align}
Here $T_\alpha(s,q)$ is the trace of the sandwiched output.
The dependence of $T_\alpha$ and $\lambda_\pm$ on the fixed
channel parameters $\gamma,N$ is suppressed.
\end{lemma}

\begin{proof}
Equation~\eqref{eq:diagonal-filtered-radius} restricts the reference
to $\sigma_s$ with $0<s<1$, including the channel endpoints.
For fixed $s$, maximizing the divergence is equivalent to maximizing
the Schatten $\alpha$-norm of the sandwiched channel output.
This norm is convex in the input, so a pure state suffices.
Up to a global phase, every pure qubit state has the form
$U_\theta\ket{\psi_q^+}$.
Phase covariance and $[U_\theta,\sigma_s]=0$ make the objective
independent of $\theta$, leaving only the optimization over $q$.
Writing $H=\sigma_s^{(1-\alpha)/(2\alpha)}$, the remaining
evaluation is
\begin{align}
 \widetilde D_\alpha(\cA_{\gamma,N}(\ketbra{\psi_q^+}{\psi_q^+})\|\sigma_s)
 &=\frac{1}{\alpha-1}\log\Tr\!\left[
 \left((\cA_{\gamma,N})_H(\ketbra{\psi_q^+}{\psi_q^+})\right)^\alpha\right]\nonumber\\
 &=\frac{1}{\alpha-1}\log
       \bigl[\lambda_+(\alpha,s,q)^\alpha
             +\lambda_-(\alpha,s,q)^\alpha\bigr].
 \label{eq:scalar-renyi-evaluation}
\end{align}
The sandwiched output has trace $T_\alpha(s,q)$ and determinant
$D_{\gamma,N}(q)[s(1-s)]^{(1-\alpha)/\alpha}$.
Its eigenvalues are therefore those
in~\eqref{eq:scalar-renyi-eigenvalues}.
Taking the reference infimum and input maximum proves
\eqref{eq:scalar-renyi-radius}.
\end{proof}

\end{document}